\documentclass[journal,a4paper]{IEEEtran}

\usepackage{ifpdf}
\usepackage{graphicx,cite}
 
\usepackage{color}
\usepackage{graphicx}

\usepackage{amsmath,amsfonts,amssymb,amsthm,mathtools}
\usepackage{hyperref}
\usepackage{graphicx,color}

\newtheorem{proposition}{Proposition}[section]
\newtheorem{theorem}[proposition]{Theorem}

\DeclarePairedDelimiter\norm{\lVert}{\rVert}

\begin{document}

\title{Euclidean Matchings in Ultra-Dense Networks}

\author{Alexander Kartun-Giles,~\IEEEmembership{Member}, Suhanya Jayaprakasam,~\IEEEmembership{Member} \\
       Sunwoo Kim,~\IEEEmembership{Member}
\thanks{The Authors are with the Wireless Systems Laboratory in the Department of Electronic Engineering, Hanyang University, Seoul, South Korea.}}

\markboth{IEEE Communications Letters,~Vol.~xx, No.~xx, ~2018}%
{}

\maketitle

\begin{abstract}
In order to study the fundamental limits of network densification, we look at the spatial spectral efficiency gain achieved when densely deployed communication devices embedded in the $d$-dimensional Euclidean space are optimally `matched' in near-neighbour pairs. In light of recent success in probabilistic modelling, we study devices distributed uniformly at random in the unit cube which enter into one-on-one contracts with each another. This is known in statistical physics as an Euclidean `matching'. Communication channels each have their own maximal data capacity given by Shannon's theorem. The length of the shortest matching then corresponds to the maximum one-hop capacity on those points. Interference is then added as a further constraint, which is modelled using shapes as guard regions, such as a disk, diametral disk, or equilateral triangle, matched to points, in a similar light to computational geometry. The disk, for example, produces the Delaunay triangulation, while the diametral disk produces a beta-skeleton. We also discuss deriving the scaling limit of both models using the replica method from the physics of disordered systems.
\end{abstract}

\IEEEpeerreviewmaketitle

\section{Introduction}\label{sec:introduction}
With a random configuration of devices in space, location-aware, directional transmission and reception entails device pairing, known in the literature as a \textit{Euclidean matching}. This consists of a binomial point process of $N$ nodes, drawn from a bounded domain, with edges added in such a way that no two edges intersect, see Fig. \ref{fig:main1}. Due to the continuity of the space, there is always a `shortest' matching, where by short we mean that the Euclidean lengths of the at most $N/2$ edges concatenate to form a rectifiable curve which minimises a length functional.

In the case of ultra-dense cellular networks, an obvious focus is on combating inter-transceiver interference by minimising a cost functional defined on an ensemble of random geometric graphs \cite{alammouri2017, alammouri20172, ge2016}. In this letter, we focus on the anticipated interference-combatting technique of directional transmission, which is becoming essential in close proximity, short-range wireless communications. The main topic of this letter is that, in our, or any communication-theoretic analogue of a Euclidean matching problem, interference appears in addition to the matching constraint. At least in the interference-free case, studying the \textit{scaling limit} of the sum of edge lengths in the shortest matching is related to the \textit{monopartite Euclidean matching problem} \cite{sicuro2017}. The \textit{bipartite} case is alternatively a one-to-one correspondence between two different types of nodes. However, in the interference-limited problem, one takes the length of the shortest matching, provides a corresponding interference-free capacity, and adjusts it according to an interference functional, which implies a dramatic augmentation of the microscopic pairings, but only, as we show, a constant factor scaling of the capacity.

For a comprehensive literature review, see both the monograph of Plummer and Lov\'{a}sz \cite{lovaszbook}, and the recent monograph of Han et al. \cite{han2017}, where the current application of `matching theory' in resource allocation problems is discussed. Topics include the \textit{stable marriage problem} in the D2D scenario \cite{han20172}, and the \textit{stable fixture model} for LTE V2X \cite{han20173}. A review of the effects of network densification can be found in Gupta, Zhang and Andrews, for example, as well as AlAmmouri, Andrews and Baccelli \cite{gupta2015, alammouri2017}. See also recent developments in location-awareness \cite{taranto2014}, directional antenna capacity gain \cite{li20112}, and in general, other, similar scenarios where spatially stochastic network models appear in wireless communications \cite{knight2017,giles2016,giles2015,koufos2016,koufos2018}.

Summarising the contents of this letter, in Section \ref{sec:model} we introduce our model, in Section \ref{sec:interference} we discuss the constraint introduced by interference, and then detail our results about the relation between the order of the data capacity in both the interference-free and interference-limited case. In Section \ref{sec:multihop} we then detail what is meant by multihop transport under this independent edge set constraint. In Section \ref{sec:otherresults}, we discuss our results in relation to previous limits on data capacity from the perspective of stochastic geometry. In Section \ref{sec:asymptotics}, we discuss rescaling in the dense limit, and the effects of a different path loss model (stretched exponential path loss), directly applicable to the ultra-dense scenario. Finally, in Section \ref{sec:conclusion} we conclude.

\section{Matching Problems on Point Processes}\label{sec:model}
For $N \in \mathbb{N}$, consider the \textit{binomial point process} $\mathcal{X}_{2N} \subset [0,1]^{d}$ of $2N$ points. Form a \textit{perfect matching} $\mathcal{M}$ of these points by assigning $N$ of the pairs active in such a way that every point is incident to exactly one active pair. Call the Euclidean lengths of these edges $d_{1},d_{2},\dots,d_{N}$. We reserve $d$ for the dimension of the hypercube. For $C,\eta>0$, we then assign each edge its own data capacity $C_{i} = \log_{2}\left(1 + d_{i}^{-\eta}\right)$ based on the Shannon-Hartley theorem, since if the received signal power is well modelled by $P_{i} = Cd_{i}^{-\eta}$
taking $\eta$ for the path loss exponent, then such a function of the edge length corresponds to the theoretical upper bound on the data capacity of a link over distance $d_{i}$, given the bandwidth and noise are set to unity. For each matching, we therefore have a length $L_{\mathcal{M}} = \sum_{i}d_{i}$
and a capacity
\begin{eqnarray}\label{e:capacity}
  C_{\mathcal{M}} =: \sum_{i=1}^{N}\log_{2}(1+d_{i}^{-\eta}).
\end{eqnarray}
The perfect matching
which minimises $L_{\mathcal{M}}$ is the \textit{shortest} perfect matching, or just \textit{shortest matching}. As an example, Fig. \ref{fig:main1} depicts solutions to the monopartite Euclidean matching problem on $N = 150$ and $300$ points. We discuss this setting in Section \ref{sec:interference}.

\section{Interference}\label{sec:interference}

\begin{figure}
  \begin{centering}
    \includegraphics[scale=0.32]{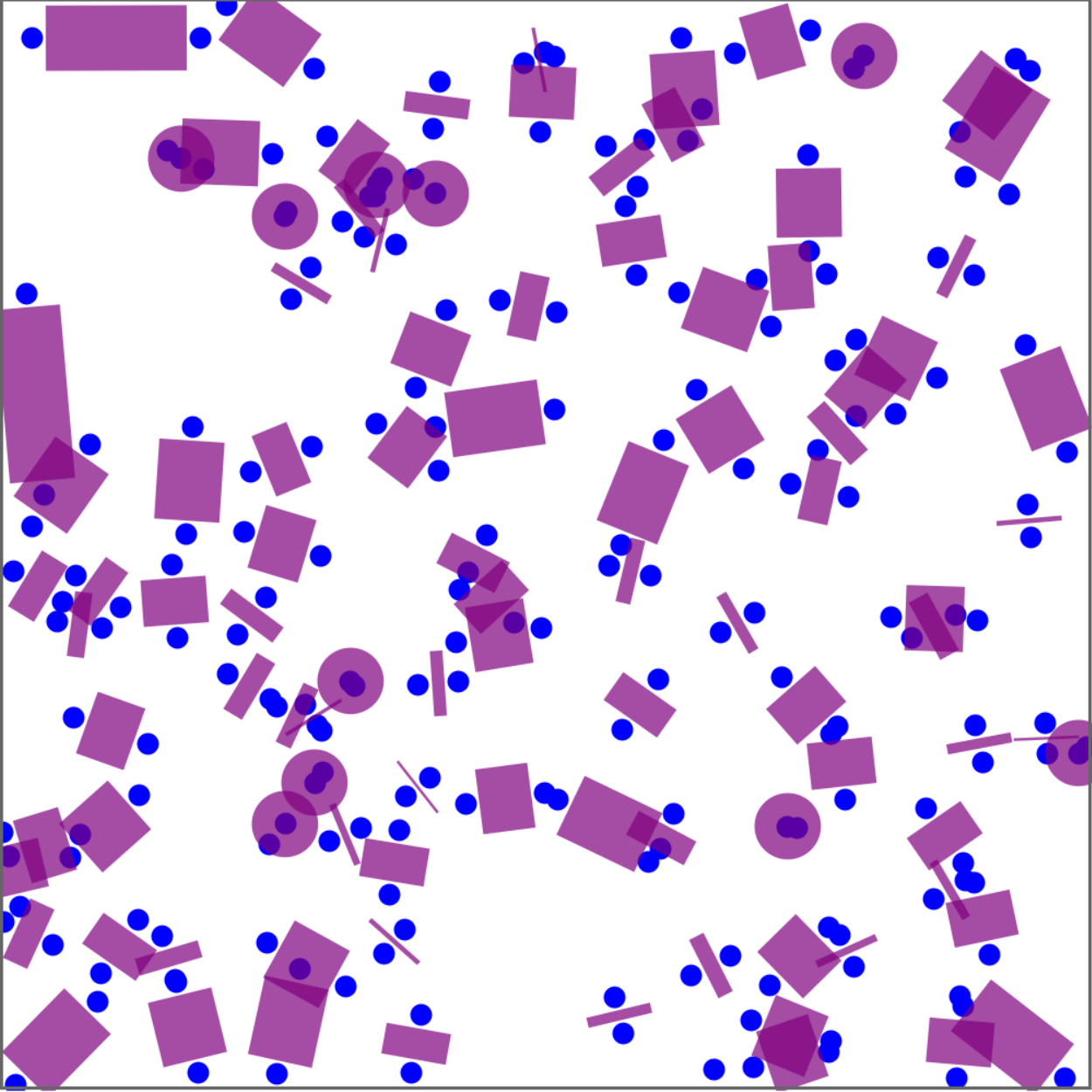} \hspace{1mm}
    \includegraphics[scale=0.32]{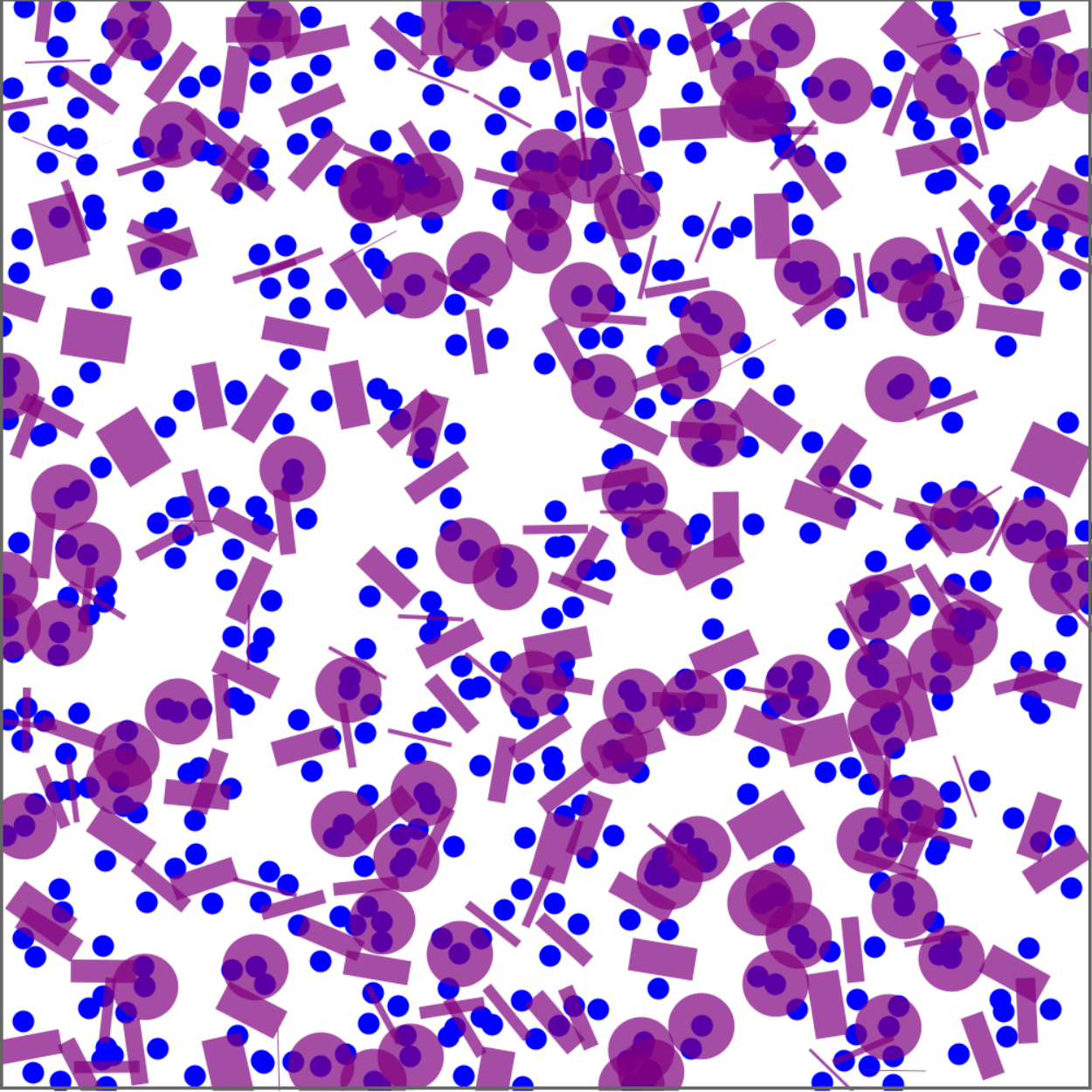} \\
    \caption{Perfect `matchings' on uniformly random points in $[0,1]^{2}$ with a density of $200$ (left) and $600$ (right) points per unit area. The purple edges represent the leading order contribuiton to an interference field created by a directional transmission.}\label{fig:main1}
    \vspace{3mm} 
    \includegraphics[scale=0.32]{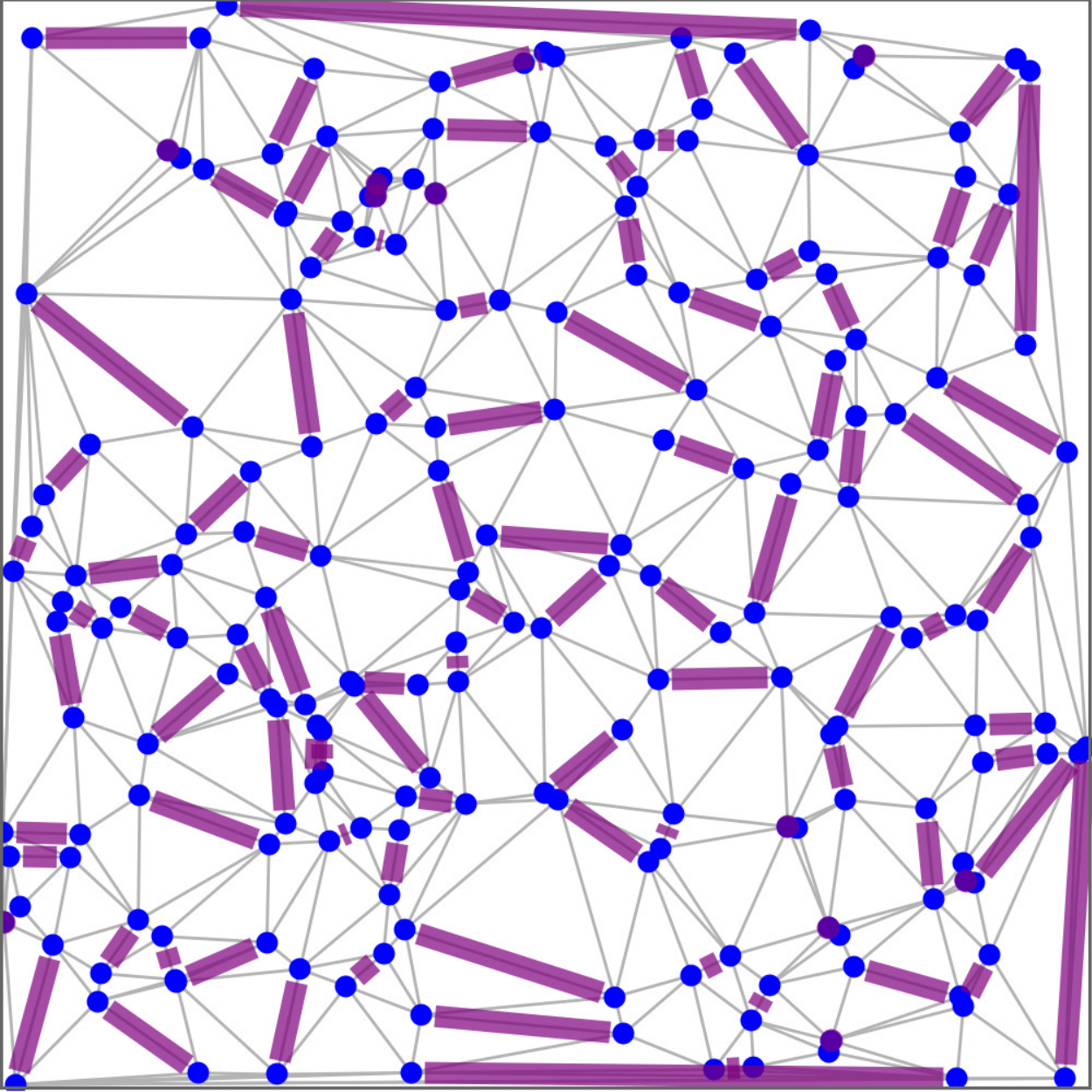} \hspace{1mm}
    \includegraphics[scale=0.32]{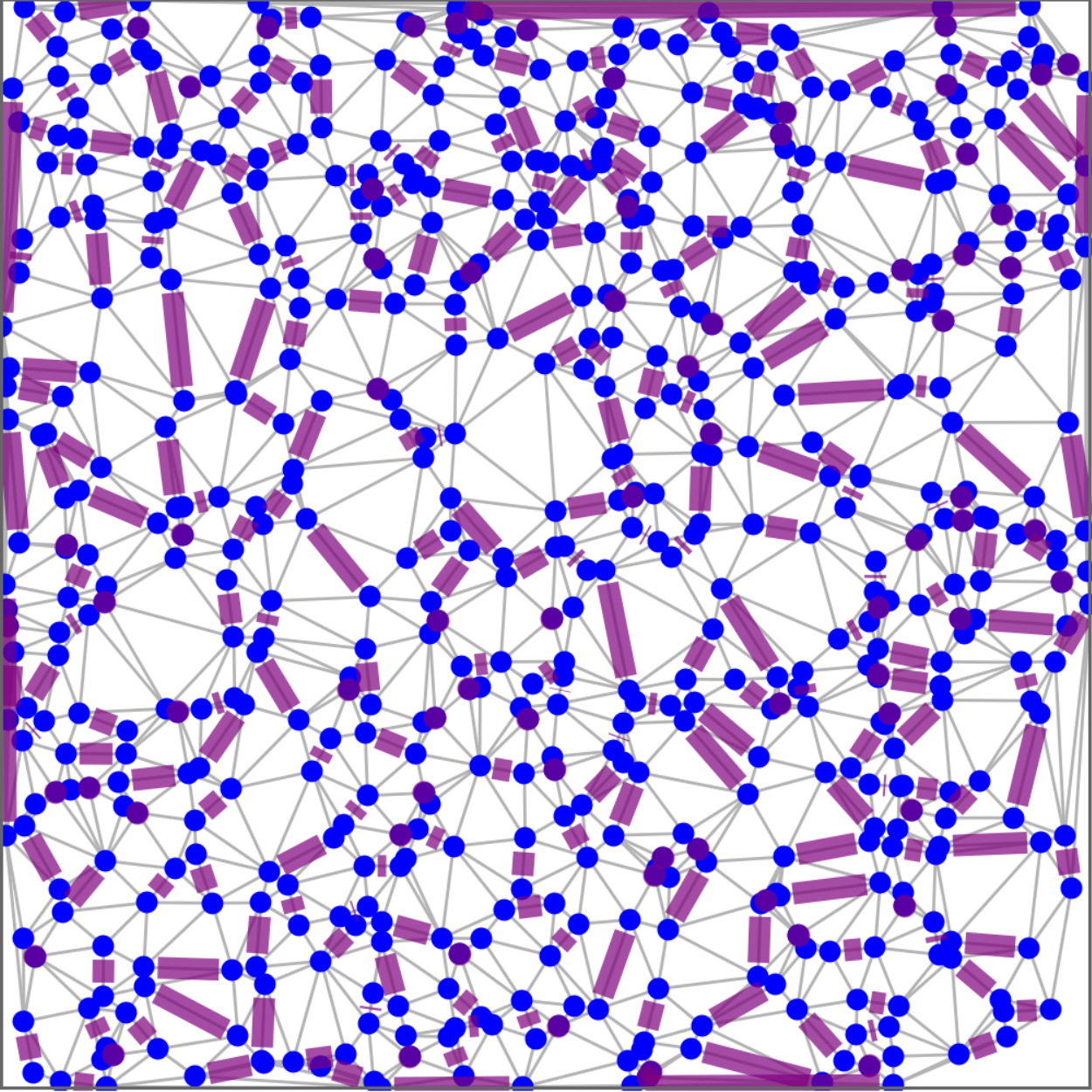} \\
  \end{centering}
  \caption{Similarly, perfect matchings on $200$ (left) and $600$ (right) points per unit area, but now selected from the unique Delaunay triangulation on the points (which is marked by the dark, framework edges). These admissible links represent disks, known as \textit{circumcircles}, which intersect pairs of points but capture no others within their interior, modelling interference guard regions.}\label{fig:main2}

\end{figure}

The majority of research on geometric tour problems, as well as the underlying network routing problems, focuses on cost functions based on the lengths of edges. However, many actual routing problems display a dominant contribution to the cost when e.g. switching paths or changing direction at a junction. A good example is Numerical Control Machining, where turns of a cutting stylus represent an important component of the objective function, as the cutter may have to be slowed in anticipation of a turn. Consider also robotic exploration, snow plowing, or street sweeping with turn penalties \cite{arkin2005}.

Interference between simultaneous transmissions is another example. To illustrate this, Fig. \ref{fig:main2} shows a short, perfect matching on $\mathcal{X}_{2N}$, but with the edges selected from the Delaunay triangulation. The matching is thus between sites of the Voronoi tessellation of the plane induced by the points. Each edge, indicated with a purple rectangle, as a consequence corresponds to a disk in the plane with a transmitter-receiver pair on its boundary, but which does not contain any other point of $\mathcal{X}_{2N}$.

This resembles models of computational geometry \cite{biniaz2016}. As a model then of interference in dense networks, we define  $G_{S}(\mathcal{X}_{2N})$ as the graph on $\mathcal{X}_{2N}$ which has an edge between two points $i,j \in \mathcal{X}_{2N}$ if and only if a scaled, translated version of a shape $S$ (known as a  \textit{homothet}) exists having both $i$ and $j$ on its boundary, and whose interior does not contain any point of $\mathcal{X}_{2N}$. If $S$ is a disk $\bigcirc$, as disucssed above, then $G_{\bigcirc}(\mathcal{X}_{2N})$ is the \textit{Delaunay triangulation} of $\mathcal{X}_{2N}$. If $S$ is an equilateral triangle $\triangledown$, then $G_{\triangledown}(\mathcal{X}_{2N})$ is the triangular distance Delaunay graph on $\mathcal{X}_{2N}$. If $S$ is a diametral disk, we instead have a $\beta$\textit{-skeleton} with $\beta=1$, i.e the \textit{Gabriel graph}, which contains the Euclidean minimal spanning tree. The beam pattern created by the transceiver pair is then symmetric.

A matching is called \textit{perfect} if all vertices are paired. A matching is called \textit{weak} if matched shapes $S$ intersect at least somewhere in the plane, and \textit{strong} if they don't. A typical matching on a Delaunay triangulation is weak. Given a set of points in the plane, the strong matching problem is to compute a strong matching on e.g. a Delaunay triangulation on those points, of maximum cardinality \cite{biniaz2017}. With respect to interference, we consider only perfect, weak matchings in this letter.

We now detail a theorem concerning how the capacity of a matching is affected by these constraints. We first define a \textit{minimum triangulation} on a set of Euclidean points as that triangulation of minimum total Euclidean length, and the \textit{Delaunay triangulation} as that triangulation of the points formed by adding edges between pairs of points which intersect a circle which itself contains no other points in its interior.

\begin{theorem}\label{t:interference}
Given the capacity of a wireless link over a distance $d_{i}$ is given by $C(d_{i}) = \log_{2}(1+d_{i}^{-\eta})$, and that there exists a perfect matching of the points of $\mathcal{X}_{2N}$, then augmentation of the matching such that it is weak with respect to the disk $\bigcirc$, that is, selected from amongst the edges of the Delaunay triangulation on $\mathcal{X}_{2N}$, almost certainly adjusts its capacity by only a constant factor.
\end{theorem}
  The minimum (length) triangulation on Euclidean points $\mathcal{X}_{2N}$ is potentially much shorter than its Delaunay counterpart. However, on random points, such as $\mathcal{X}_{2N}$, the ratio of the total Euclidean length of the Delaunay triangulation to the minimum triangulation is almost always $\mathcal{O}(1)$. See \cite{chang1984} for a proof of this. Their argument implies the Delaunay triangulation is a ``good'' approximation to minimum triangulation in this case. The length of the shortest matching on uniformly random points is similarly well approximated by the shortest matching on the Delaunay triangulation, though we do not demonstrate this rigourously. Thus, moving from an interference-free to interference-limited matching does not affect the order of the capacity. One would have to complete, however, a rigorous proof of the $\mathcal{O}(1)$ length difference between the minimum length matching on the minimum triangulation, and the minimum matching on the points themselves, in order for a proof to be rigourous. We thus only argue heuristically for Theorem \ref{t:interference}. 

\begin{theorem}\label{t:gupta}
Given the capacity of a wireless link over a distance $d_{i}$ is given by $C(d_{i}) = \log_{2}(1+d_{i}^{-\eta})$, and that there exists a perfect matching of the points of $\mathcal{X}_{2N}$ where the Euclidean lengths of the edges in the matching are each of order $N^{-1/d}$, then the one-hop throughput capacity of the interference-limited network, based on a strong matching on a disk $\bigcirc$, is $\mathcal{O}\left(\log N \right)$.
\end{theorem}
\begin{proof}
  Considering a matching where devices are paired with their near-neighbours. We have that the length $L_{\mathcal{M}}$ is of order $N^{1-1/d}$, since there are $N$ edges each with length of order $N^{-1/d}$ \cite[Section 5.1]{aldous2001}. Wireless links have a corresponding limiting capacity of order
\begin{eqnarray}
C_{i} = \mathcal{O}\left(\log_{2}\left(1+\left(\frac{1}{N^{1/d}}\right)^{-\eta}\right)\right).
\end{eqnarray}
Since $N$ links can transmit simultaneously, the one-hop throughput capacity of the interference-free network is
\begin{eqnarray}
\mathcal{O}\left(\frac{1}{N}\sum_{i=1}^{N}\log_{2}\left(1+N^{\eta/d}\right)\right) = \mathcal{O}\left(\log N\right).
\end{eqnarray}
By Theorem \ref{t:interference}, the interference-limited capacity, based on finding a weak matching of $\mathcal{X}_{2N}$ on the disk $\bigcirc$, is only a constant factor larger than this, and so its order remains, at least in our model, $\mathcal{O}\left(\log N\right)$.
\end{proof}

Under a more realistic model of interference, it may be possible to increase the capacity by choosing a longer matching. We defer this question to a later study of interference in the monopartite case.

\section{Multiple-Hop Capacity}\label{sec:multihop}
A simple example of a Markov chain on the space of matchings with local update rules was originally proposed by Diaconis et al. \cite{diaconis2001}, known as the \textit{switch chain}. The endpoints of a pair of disjoint edges are swapped when moving between adjacent states \cite{dyer2017}. This could be used to model multi-hop delays, which reduce capacity. This would require a central computer to list all perfect matchings of the communication graph, which is P-complete \cite{valiant1979}, even in the simpler bipartite case i.e. this task is intractable without a distributed algorithm \cite{bertsekas1988}.

\section{Relation to Previous Capacity Limits}\label{sec:otherresults}
We briefly highlight the relation to the $\sqrt{n}$ capacity growth discussed in detail in both AlAmmouri et al. \cite{alammouri20172}, and earlier by both Gupta and Kumar \cite{gupta2000} and Franceschetti et. al \cite{franceschetti2009}. We state that, assuming our shape-based inteference model holds for all $n$, then the network capacity of dense stations using directional transmission to near-neighbour receivers scales as $\mathcal{O}(n \log n)$. In the mesh network case with $\mathcal{O}(\sqrt{n})$ hops between source and destination, the capacity limit is scaled to $\mathcal{O}(\sqrt{n} \log n)$. We differ from the fundamental $\sqrt{n}$ limits firstly due to the use of directional transmission, which isolates a communication pair from the rest of the network in a way which appears to allow capacity to scale in this way, and also our relatively elementary interference model, which may become unrealistic for sufficiently large $n$ (though potentially beyond that which is expected of ultra-dense deployment). The capacity per node in the multi-hop case with uniformly random destinations vanishes in the dense limit, which agrees with previous results. We defer further treatment of this point to a more detailed study.

\section{Asymptotics and the Replica Method}\label{sec:asymptotics}
We need to rescale the edge lengths of the matching in order to study the asymptotic case where $N \to \infty$. We detail this procedure here, and discuss the replica method which is used in similar problems in statistical physics.
\begin{proposition}\label{p:mezard}
  Label the points $x_{t} \in \mathcal{X}_{2N}$, $t \in \left[1,2N\right]$. Take $\norm{x-y}$ to be the Euclidean distance between $x,y \in \mathcal{X}_{2N}$. Also, construct the $2N \times 2N$ symmetric matrix $\left(c\left(i,j\right)\right)$, where $c(i,j)=\norm{x_{i}-x_{j}}$. Then, with $\pi$ a permutation of the points of $\mathcal{X}_{2N}$, there exists a $C_{\eta,d} > 0$ such that
  \begin{eqnarray}\label{e:capacity}
\lim_{N \to \infty}\frac{1}{N}\sum_{i=1}^{N}\log_{2}\left(1+\left(N^{1/d}d_{i}\right)^{-\eta}\right) \to C_{\eta,d}.
  \end{eqnarray}
\end{proposition}
\begin{proof}
  In the monopartite case where all vertices $x_{t} \in \mathcal{X}_{2N}$, $t \in \left[1,2N\right]$ are of the same type, the set of perfect matchings of $\mathcal{X}_{2N}$ is in one to one correspondence with the set of permutations $\Pi$ of $[1,2N]$. If a permutation $\pi \in \Pi$ has elements $\pi(1) \dots \pi(2N)$, the matching has edge set $E$ given by $N$ sequential non-overlapping pairs $\left(\pi(1),\pi(2)\right), \dots, \left(\pi(2N-1),\pi(2N)\right)$, and the corresponding communication graph is $G_{\mathcal{M}}=(\mathcal{X}_{2N},E)$.
  A perfect matching in the bipartite case is similar, taking the form of a one-to-one correspondence between two non-overlapping point sets.

  With any permutation of the points, we have $2N(2N-1)$ inter-point distances given by the coefficients $c(i,j)$.
To simplify things, we remove the spatial correlation of edge weights by taking them to be i.i.d. random reals. In this setting, the $c(i,j)$ probability density
$\rho_{c(i,j)}(l) \sim l^{d-1} \text{ as } l \downarrow 0$. To see this, simply imagine a disk of radius $l$ around a point of $\mathcal{X}_{2N}$, and compare its circumference.

In a manner corresponding to the argument of M\'{e}zard and Parisi in \cite[Eqs. 22 and 23]{mezard1985}, we can see from Eq. \ref{e:capacity} that the expected capacity of the shortest matching denoted $\mathbb{E}C_{\mathcal{M}}$ grows as $\mathcal{O}\left(N \log N\right)$. If we rescale the problem by this expectation, and instead study $C^{\prime}_{\mathcal{M}}=(N\log N)^{-1}C_{\mathcal{M}}$, the now rescaled capacity should converge to an unkown constant, since the deviation of the finite capacity from its expectation is bounded \cite{aldous2001}. In other words, we now look at
  \begin{eqnarray}\label{e:capacity2}
C^{\prime}_{\mathcal{M}}=\lim_{N \to \infty}\frac{1}{N}\sum_{i=1}^{N}\log_{2}\left(1+d_{i}^{\prime -\eta}\right)
  \end{eqnarray}
where $d_{i}^{\prime}=N^{1/d}d_{i}$. This \textit{independent link model} is motivated by long standing  attempts in statistical physics, particularly the theory of the magnetic alloys known as \textit{spin glasses}, to try and fit so called \textit{mean field models} to the Euclidean matching problem \cite{aldous2001,mezard1988}, since the total length on correlated edges are mathematically difficult to analyse, and these more tractable models can be rather involved in their own right. In the interference-limited case, we may be able to study phase transitions in the capacity as we continuously scale certain parameters, such as shape parameters of guard regions $S$ around transceiver pairs.
\end{proof}

Similarly, the limit of the rescaled length of a perfect matching on the Delaunay triangulation also converges, but to a different limit. We defer these details to a later study of interference in monopartite Euclidean matchings. Also, going on to evaluate these constants in e.g. Eq. \ref{e:capacity2} appears to be an intractable non-linear optimisation problem \cite{berstein2008}.

As a more tractable model, consider a \textit{stretched exponential} path loss function, rather than the power law of Eq. \ref{e:capacity} \cite{alammouri2017}. With fitting parameters $\alpha,\beta>0$, signal powers attenuates instead over a distance $r$ as $\exp\left(-\alpha r^{\beta}\right)$. Consequently, as $ d_{i} \downarrow 0$, we can take advantage of a Shannon capacity which follows a power law in the link distance, since with $A(\alpha)$ some function which does not depend on $d_{i}$, then
\begin{eqnarray}\label{e:exploss}
 \log_{2}\left(1 + e^{- \alpha d_{i}^{\beta}}\right) \to 1 - A\left(\alpha\right) d_{i}^{\beta}
\end{eqnarray}
via a Maclaurin expansion.

Also, due to Mezard and Parisi we know that, asymptotically, the shortest length $L_{d}$ in the independent link model of the rescaled monopartite Euclidean matching problem is given by
\begin{equation}\label{e:mezard}
L_{d} = 2d\int_{-\infty}^\infty G(l) e^{-G(l)} \mathrm{d}l
\end{equation}
where $G$ is obtained numerically from the following equation
\begin{equation}
G(l) = \frac{2}{(d-1)!} \int_{-l}^\infty (l+y)^{d-1} e^{-G(y)} \mathrm{d}y
\end{equation}
with $-\infty < l < \infty$, see \cite[Section 5]{aldous2001}. Due to the simplification of Eq. \ref{e:exploss}, we need only replace Eq. \ref{e:mezard} with Eq. \ref{e:capacity}, and use the \textit{replica method} to find an expression for the limiting capacity, following the detailed calculations of those authors, see e.g. \cite{mezard1985}. This sophisticated technique will be deferred to a later work, as well as the similar situation in the monopartite case.

\section{Conclusions}\label{sec:conclusion}
We have shown how a model of Euclidean matchings can be used to derive scaling limits for the data capacity of interference-limited ultra-dense networks. We discussed how, if a matching of minimum Euclidean length is arranged, such a device pairing will display a maximal data capacity in the interference-free case. When edges are assigned capacities based on their Euclidean length, but must instead be selected from the Delaunay triangulation on the points, the network is able to maximise its capacity as well as minimise interference, given the model's assumptions.

We proved that, with $\eta$ the path loss exponent, given the capacity of a wireless link over a distance $d_{i}$ is given by $C(d_{i}) = \log_{2}(1+d_{i}^{-\eta})$, and that there exists a perfect matching of the points of $\mathcal{X}_{2N}$ where the Euclidean lengths of the edges are each of order $N^{-1/d}$, then the one-hop capacity of the monopartite network is $\mathcal{O}(N\log{N})$. Finally, we discussed a network limit, whose capacity is related to a similar problems in the statistical physics of disordered systems. We believe ultra-dense networks can in the future be further successfully entwined with these important models and ideas.

\section*{Acknowledgements}
This work is supported by the Samsung Research Funding and Incubation Center of Samsung Electronics under Project Number SRFC-IT-1601-09. The first author also acknowledges support from the EPSRC Institutional Sponsorship Grant \textit{Random Walks on Random Geometric Networks}.



\end{document}